\documentclass[english,thm-restate]{lipics-v2021}
\usepackage{booktabs}
\usepackage{xspace}
\usepackage{comment}
\usepackage{thmtools}
\usepackage{booktabs}
\usepackage{todonotes}
\nolinenumbers
\hideLIPIcs  

\usepackage{microtype} 

\newcommand\polylog{{\rm polylog}}
\graphicspath{ {figures/} }

\newcommand {\R} {\mathbb{R}\xspace}

\newtheorem*{conjecture*}{Conjecture}

\newcommand{\enumit}[1]{\textcolor{darkgray}{\sffamily\bfseries\upshape\mathversion{bold}#1}}

\newcommand{\mkmcal}[1]{\ensuremath{\mathcal{#1}}\xspace}

\newcommand{\F}{\mkmcal{F}}

\newcommand{\T}{\mkmcal{T}}

\newcommand{\POld}{\ensuremath{\tilde{P}}}
\newcommand{\mOld}{\ensuremath{\tilde{m}}}
\newcommand{\SOld}{\ensuremath{\tilde{S}}}
\newcommand{\nOld}{\ensuremath{\tilde{n}}}
\newcommand{\tri}{\ensuremath{\blacktriangle}}
\newcommand{\Tr}{\ensuremath{\mathfrak{T}}\xspace}
\newcommand{\Vor}{\ensuremath{\mathit{VD}}\xspace}

\newcommand{\etal}{et al.\xspace}

\newcommand{\dynamicsp}{\textsc{Dynamic Shortest Path}\xspace}
\newcommand{\conequery}{\textsc{Cone Query}\xspace}

\def\polylog{\operatorname{polylog}}

\title{Nearest Neighbor Searching in a Dynamic Simple Polygon}

\author{Sarita de Berg}{Department of Information and Computing Sciences, Utrecht University, The Netherlands}{s.deberg@uu.nl}{}{}

\author{Frank Staals}{Department of Information and Computing Sciences, Utrecht University, The Netherlands}{f.staals@uu.nl}{}{}

\authorrunning{S. de Berg and F. Staals} 



\Copyright{Sarita de Berg and Frank Staals} 

\authorrunning{S. de Berg and F. Staals}
\ArticleNo{100}

\begin{document}

\maketitle

\begin{abstract}
    In the nearest neighbor problem, we are given a set $S$ of point sites that we want to store such that we can find the nearest neighbor of a (new) query point efficiently. In the dynamic version of the problem, the goal is to design a data structure that supports both efficient queries and updates, i.e. insertions and deletions in $S$. This problem has been widely studied in various settings, ranging from points in the plane to more general distance measures and even points within simple polygons. When the sites do not live in the plane but in some domain, another dynamic problem arises: what happens if not the sites, but the domain itself is subject to updates?

    Updating sites often results in local changes to the solution or
    data structure, while updating the domain may incur many global
    changes. For example, in the closest pair problem, inserting a
    point only requires us to check if this point is in the new
    closest pair, while updating the domain might change the distances
    between most pairs of points in our set. Presumably, this is the reason that this form of dynamization has received much less attention. Only some basic problems, such as shortest paths and ray shooting, have been studied in this setting.
    
    Here, we tackle the nearest neighbor problem in a dynamic simple polygon. We allow insertions into both the set of sites and the polygon. An insertion in the polygon is the addition of a line segment starting at the boundary of the polygon. We present a near-linear size --in both the number of sites and the complexity of the polygon-- data structure with sublinear update and query time. This is the first nearest neighbor data structure that allows for updates to the domain.
\end{abstract}

\section{Introduction}

The nearest neighbor problem is a fundamental problem in computer
science: given a set of~$n$ point sites $S$ in some domain $P$, the
goal is to store $S$ so that given a query point $q \in P$ we can
efficiently find the point in $S$ closest to
$q$~\cite{Chan10,Kaplan17}. In case $P$ is the two-dimensional
Euclidean plane, such queries can be answered in $O(\log n)$ time by
computing the Voronoi diagram of $S$, and storing it for efficient
point location queries. In many applications, the domain however may
contain obstacles (e.g. buildings, cars, rivers, lakes etc.) that are
impassible. Hence, we may want to model $P$ as a simple polygon, or
even a polygonal domain (i.e. a polygon with holes), and measure the
distance between two points by the length of their shortest obstacle
avoiding path. In this setting, it is also possible to build the
Voronoi diagram. Aronov~\cite{aronov1989geodesic} showed that the so
called \emph{geodesic Voronoi diagram} in a simple polygon with $m$ vertices has
complexity $O(n+m)$, and gave an $O((n+m)\log(n+m)\log m)$ time algorithm to construct
it. This then allows for $O(\log (n+m))$ time nearest neighbor
queries. Finding an optimal algorithm to construct the geodesic Voronoi diagram remained elusive for quite a while. Only a few years ago,
after a series of improvements~\cite{
  liu20nearl_optim_algor_geodes_voron,oh20voron_diagr_moder_sized_point,papadopoulou1998geodesic}, Oh showed an optimal
$O(m + n\log n)$ time algorithm~\cite{Geodesic_voronoi_diagram}. When
$P$ is a polygon with holes, the Voronoi diagram has the same complexity, and can be computed in $O((n+m)\log(n+m))$
time~\cite{hershberger1999sssp}. 

Nearest neighbor searching is a \emph{decomposable search problem},
meaning that to find the nearest neighbor in the set $A \cup B$, we
can compute the nearest neighbor in set $A$, separately compute the
nearest neighbor in set $B$, and return the closest of these two
candidates~\cite{bentley1980decomposable}. Therefore, supporting efficient insertions of
new point sites in $\R^2$ (while still answering queries in
polylogarithmic time) is relatively simple using, for example, the logarithmic method~\cite{Overmars83}. However,
supporting both insertions and deletions in polylogarithmic time,
turned out to be quite challenging. Chan~\cite{Chan10}  presented the
first such results, achieving $O(\log^2 n)$ query time while
supporting insertions in $O(\log^3 n)$ and deletions in $O(\log^6 n)$
amortized expected time. Kaplan \etal~\cite{Kaplan17} and later
Chan~\cite{Chan20_dynamic_nn} further improved the update times to
$O(\log^2 n)$ for insertions, and $O(\log^4 n)$ for deletions. Kaplan \etal also generalized
these results to other (constant complexity) distance
functions~\cite{Kaplan17}. Agarwal, Arge, and
Staals~\cite{dynamic_geodesic_nn} showed that these techniques can be extended to efficiently answer nearest neighbor queries in a
simple polygon $P$. De Berg and Staals even presented a data structure
that can report efficiently report the $k$-nearest neighbors, for some
query parameter $k \in [n]$, in this setting~\cite{Dynamic_knn}.

The above results thus support applications in which the set of input
points may change over time. However, in such applications, the
obstacles are also likely to change. Unfortunately, none of the
approaches above support updates to the domain $P$, e.g. by inserting
or deleting an obstacle. Even more generally, there is relatively
little work in which the input domain~$P$ may change while maintaining
some non-trivial structure on $P$ or $S$. Results of this type are
mostly restricted to general purpose range or intersection
searching~\cite{agarwal93applic_new_space_partit_techn,
  dynamic_ray_shooting} or point
location~\cite{dynamic_point_location}. The most prominent work that
we are aware of that specifically maintains some data structure on top
of a dynamic polygonal domain is the work by Goodrich and
Tamassia~\cite{dynamic_ray_shooting_sp}. They support ray shooting and
shortest path queries in a dynamic connected planar subdivision $\hat{P}$ in
$O(\log^2 n)$ time. Ishaque and T{\'{o}}th show how to maintain the
(geodesic) convex hull of $S$ in $\hat{P}$ while supporting the insertion of
barriers (non-crossing segments such that all faces are simply connected) and site deletions in
$O(\polylog (n+m))$
time~\cite{ishaque14relat_convex_hulls_semi_dynam_arran}. Oh and Ahn
generalize this result to supporting both insertions and deletions (of
barriers and sites)~\cite{dynamic_convex_hull}. They achieve update
times of $\tilde{O}((n+m)^{2/3})$, where the $\tilde{O}(\cdot)$ notation
hides logarithmic factors in $n$ and $m$. A recent result of Choudhury
and Inkulu~\cite{dynamic_visibility} shows how to maintain the
visibility graph inside a dynamic simple polygon $P$ using $\tilde{O}(\Delta)$ time per update,
where $\Delta$ is the number of changes.

  \begin{figure}
      \centering
      \includegraphics{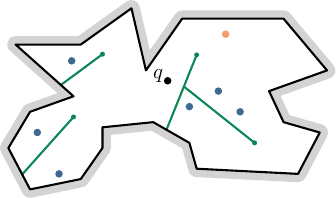}
      \caption{After inserting the green segments, the orange site is the nearest neighbor of $q$.}
      \label{fig:example-intro}
  \end{figure}

\subparagraph{Challenges.} A possible explanation for the lack of work
in this dynamic domain setting is that the problems seem extremely challenging: even just adding a
single line segment to the boundary of the polygon may significantly
change the distances between the sites in the polygon. This
severely limits what information (i.e. which distances) the data
structure can store. An insertion may also
structurally change the triangulation of the polygon (which is the
foundation upon which, e.g., the data structure of Agarwal, Arge,
and Staals~\cite{dynamic_geodesic_nn} is built).

\subparagraph{Our results.} We overcome these challenges, and present
the first dynamic data structure for the geodesic nearest neighbor
problem that allows updates to both the sites $S$ and the
simple polygon $P$ containing $S$. Our data structure achieves
sublinear query and update time and allows the following updates on $P$ and $S$ (Fig.~\ref{fig:example-intro}): \enumit{1.} inserting a segment (barrier) between points $u$ and $v$,
  where $u$ on the boundary~$\partial P$ of $P$ and $v \in P$ an arbitrary point, provided that the interior of this segment does
  not intersect $\partial P$, and \enumit{2.} inserting a site $s \in P$ into $S$.

This type of barrier insertion is similar to the insertion operation
considered by Ishaque and
T{\'{o}}th~\cite{ishaque14relat_convex_hulls_semi_dynam_arran} and Oh
and Ahn~\cite{dynamic_convex_hull}. Note that in case $v$ also lies on
the boundary of $P$, this actually splits the polygon into two
subpolygons, which can no longer interact. Hence, we can essentially
treat these subpolygons independently. In the remainder of the paper,
we thus focus on the more interesting case in which $v$ lies strictly
in the interior of $P$. This also slightly simplifies our
presentation, since $P$ then remains an actual simple polygon.

Our data structure uses $O(n(\log\log n + \log m) + m )$ space,
supports the above updates in amortized $\tilde{O}(n^{3/4}+m^{3/4})$ time, and
answers queries in expected $\tilde{O}(n^{3/4}+n^{1/4}m^{1/2})$
time.
While these query and update times are still far off from the
polylogarithmic query and update times for a static domain, they are
comparable to, for example, the result of Oh and Ahn for maintaining
the geodesic hull in a dynamic polygon~\cite{dynamic_convex_hull}.


\subparagraph{Organization.} Our data structure essentially consists
of two independent parts, both of which critically build upon balanced
geodesic triangulations~\cite{Balanced_geodesic_triangulation}. So, we
first review some of their useful properties in
Section~\ref{sec:preliminaries}. We then give an overview of the data
structure in
Section~\ref{sec:overview}. Section~\ref{sec:dynamic_shortest_path}
describes the first part of our data structure, which can be regarded
as an extension of the dynamic two-point shortest path data structure
of Goodrich and Tamassia~\cite{dynamic_ray_shooting_sp}. Our main
technical contribution is in Section~\ref{sec:conequery}, which
describes the second main part: a (static) data structure that can
answer nearest neighbor queries in a subpolygon $Q$ of $P$ that can be
specified at query time. This result may be of independent
interest. For this data structure, we need to generalize some of the
results of Agarwal, Arge, and Staals~\cite{dynamic_geodesic_nn}, which
we present in Section~\ref{sec:forest_ds_shortest_path}. Finally, in Section~\ref{sec:concluding_remarks}, we provide some concluding remarks.

\section{Preliminaries}\label{sec:preliminaries}

Let $P$ be a simple polygon with $m$ vertices, let $\partial P$ be the
boundary of $P$, and let $S$ be a set of $n$ point sites inside $P$.
We denote the shortest path that is fully contained within $P$ between
two point $p,q \in P$ by $\pi_P(p,q)$, or simply $\pi(p,q)$ if the
polygon is clear from the context. Similar to earlier papers about
geodesic Voronoi diagrams, we will use the general position assumptions
that no point is (ever) equidistant to four
sites in $S$, and no vertex of $P$ is equidistant to three
sites. 

A \emph{geodesic triangulation} of $P$ is a triangulation-like
decomposition of~$P$ into geodesic triangles. A \emph{geodesic
  triangle} with corners $u,v,w \in P$ is formed by the three shortest paths
$\pi(u,v)$, $\pi(v,w)$, $\pi(w,u)$. In general, such a
geodesic triangle consists of a simple polygon bounded by three
concave chains, and three polygonal chains emanating from the
vertices where the the concave chains join. The interior of the simple
polygon region is called the \emph{deltoid} region, and the polygonal
chains are called the \emph{tails} of the geodesic triangle. Note that
the deltoid region may be empty, for example if $w$ lies on
$\pi(u,v)$.

Similar to a triangulation, a geodesic triangulation decomposes $P$
into $O(m)$ geodesic triangles with corners that are vertices of $P$
whose boundaries do not cross. This triangulation induces a dual tree
$\T$, where two nodes are connected if their respective geodesic
triangles share two corners. We root this tree at a node whose
geodesic triangle shares a side with an edge of $P$. A geodesic
triangulation is called \emph{balanced} if the height of this tree is
$O(\log m)$. A balanced geodesic triangulation can be constructed in
$O(m)$ time~\cite{Balanced_geodesic_triangulation}.

\subparagraph{A dynamic shortest path data structure.} Here we shortly discuss the data structure of Goodrich and Tamassia for dynamic planar subdivisions~\cite{dynamic_ray_shooting_sp}.
This data structure supports both shortest-path and ray-shooting queries in $O(\log^2 m)$ time. Updates to the subdivision, such as inserting or deleting a vertex or edge, also take $O(\log^2 m)$ time.
The main component is a balanced geodesic triangulation with dual tree $\T$. A secondary dynamic point location data structure~\cite{dynamic_point_location} is build for this geodesic triangulation. A tertiary data structure stores, for each deltoid $\Delta$, the three concave chains bounding $\Delta$ in balanced tree structures. To support shortest path queries, they store for each node in $\T$ the tail of its geodesic triangle that is not shared by its parent in a balanced tree structure. The size of the data structure is $O(m)$. Using rotations, insertions and deletions can be performed in $O(\log^2 m)$ time.

\section{Overview of the data structure}
\label{sec:overview}

Our dynamic nearest neighbor data structure for a dynamic simple polygon consists of two independent data structures: a \dynamicsp data structure, which remains up-to-date, and a \conequery data structure, which is static and is rebuild after a number of updates. We denote by $P$ the current polygon with $m$ vertices and by $S$ the current set of~$n$ sites in $P$. Furthermore, we denote by $\POld$ and $\mOld$ the polygon and number of vertices and by $\SOld$ and $\nOld$ the set of sites and its size at the last rebuild of the \conequery data structure. Let $k$ be the number of updates (insertions) since the last rebuild, i.e. $k = m + n - \mOld - \nOld$.

The main purpose of the two data structures is the following. When
given a query point~$q$, we first use the \dynamicsp data structure to compute the distance from $q$ to each of the newly inserted sites, i.e. in $S \setminus \SOld$. This gives us the nearest neighbor of $q$ among these sites.
To find the nearest neighbor of $q$ among $\SOld$, we first use the \dynamicsp data structure to find
$O(k)$ regions in $P$, where each region $R$ is paired with an
apex point $v_R$ on $\partial R$. Each regions has the property that
the shortest path $\pi_P(p,q)$ from a point $p \in R$ to $q$ is given
by $\pi_{\POld}(p, v_R) \cup \pi_P(v_R,q)$. This means that the
nearest neighbor of $q$ in $P$ of the sites $\SOld \cap R$ is the same as
the nearest neighbor of $v_R$ in $\POld$ of the same sites $\SOld \cap
R$. 

This is where the \conequery query data structure comes into
play. This data structure allows us the answer such a nearest neighbor
query for some region $R$ with apex $v_R$ in $\POld$. After querying
the \conequery data structure for each of the $O(k)$ regions generated
by the \dynamicsp data structure, we return the nearest neighbor
of~$q$ among $\SOld$, which is the closest of these $O(k)$ sites. Next, we define the
queries for the \conequery data structure more precisely. We refer to
these queries as \emph{bounded nearest neighbor queries}.

\begin{figure}
    \centering
    \includegraphics[page=2]{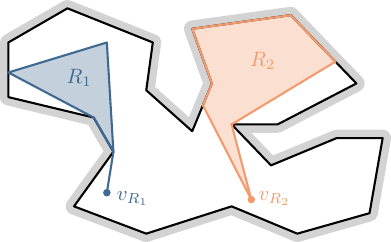}
    \caption{Two bounded nearest neighbor queries: a geodesic triangle ($R_1$) and a geodesic cone~($R_2$).}
    \label{fig:region_types}
\end{figure}

\begin{definition}[Figure~\ref{fig:region_types}]\label{def:bounded_nn}
    A \emph{bounded nearest neighbor query} $(R,v_R)$ in $\POld$ consists of a region $R$ and an point $v_R \in \partial R$ and asks for the nearest neighbor of $v_R$ in $\SOld \cap R$, where $R$ is of one of the following two types:
    \begin{itemize}
        \item a \emph{geodesic triangle}: a region bounded by three
          shortest paths in $\POld$, or
        \item a \emph{geodesic cone}: a region bounded by two shortest
          paths in $\POld$ and part of $\partial \POld$.
    \end{itemize}
\end{definition}

The \dynamicsp data structure that finds these geodesic triangle and geodesic cone regions is based on the data structure of Goodrich and Tamassia~\cite{dynamic_ray_shooting_sp} for ray shooting and shortest path queries in a dynamic planar subdivision. We augment this data structure to distinguish between ``old'' and ``new'' vertices, i.e. the vertices in $\POld$ and in $P \setminus \POld$. In Section~\ref{sec:dynamic_shortest_path}, we discuss this data structure in more detail, and obtain the following result.

\begin{restatable}{theorem}{dynamicSP}\label{thm:dynamic_sp}
    A \dynamicsp data structure maintains a dynamic simple polygon with $O(\log^2m)$ update time using $O(m)$ space such that for a query point $q$ a set of $O(k)$ bounded nearest neighbor queries that together find the nearest neighbor of $q$ in $P$ among $\SOld$ in $P$ can be computed in $O(k \log^2 m)$ time.
\end{restatable}

The \conequery data structure is the main data
structure we present. We give a short overview of the data structure
here, illustrated in Figure~\ref{fig:overview-conequery}, and discuss the details and proofs in
Section~\ref{sec:conequery}. The data structure is based on a balanced
geodesic triangulation $\T$, augmented with three auxiliary data structures. We call the geodesic
triangulation the first level data structure, and the three auxiliary
data structures the second level data structures. The second level
data structures are the following. First, for each node in $\T$, we store a
partition tree~\cite{Chan_partition_tree} on the sites that lie in the corresponding deltoid. Furthermore, we store several third
level data structures for each node of the partition tree. This allows us to
find the nearest neighbor of $v_R$ for (a subset of) the
sites contained in the deltoid. Second, for each node in~$\T$, we store the sites that lie on the boundary of the geodesic triangle in a binary search tree. This allows us the find the nearest neighbor of $v_R$ for a subset of the sites that lie on the boundary of the geodesic triangle.
Third, for each edge in $\T$, which corresponds to a shortest path between two points on $\partial \POld$, we store a nearest neighbor data structure that allows
us to compute the nearest neighbor of $v_R$ among the
sites in one of the subpolygons defined by this shortest path, provided that $v_R$
lies in the other subpolygon. We show that we can answer a bounded nearest neighbor query by querying $O(\log \mOld)$ of these level two data structures. This results in the following theorem.

\begin{figure}
    \centering
    \includegraphics{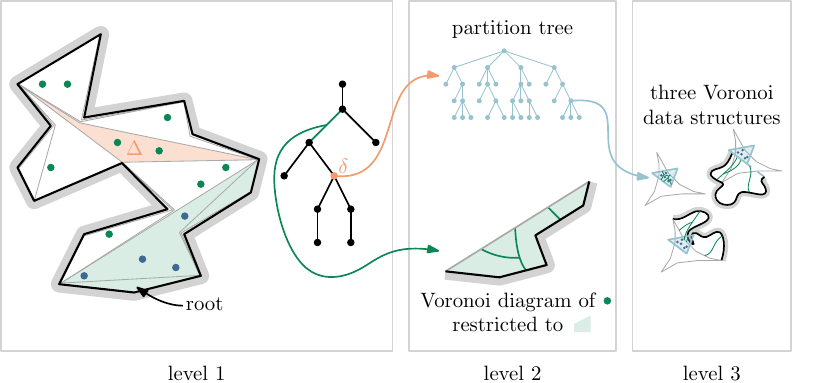}
    \caption{Overview of the \conequery data structure. For each node in the geodesic triangulation (level 1) we store a partition tree and for each edge we store some restricted Voronoi diagram (level~2). For each node in the partition tree we store three Voronoi based data structures (level 3). Nodes of the geodesic triangulation whos deltoid is empty are omitted.}
    \label{fig:overview-conequery}
\end{figure}

\begin{restatable}{theorem}{THMconeQuery}\label{thm:conequery}
    For a static polygon $\POld$ with $\mOld$ vertices and a set $\SOld$ of $\nOld$ sites in $\POld$, the \conequery data structure can be constructed in $O(\mOld \log \log \nOld + 
\nOld \log \nOld \log \log \nOld + 
\nOld \log \nOld \log^2 \mOld))$ time and $O(\nOld \log \log \nOld + \nOld \log \mOld + \mOld)$ space, such that a bounded nearest neighbor query can be answered in $O(\sqrt{\nOld} \log^{3/2} \mOld + \log \nOld \log^2 \mOld )$ expected time. 
\end{restatable}

To obtain our final result, we rebuild the \conequery data structure once $k$ becomes larger than $\nOld^{1/4} + \mOld^{1/4}$. This results in update and query times that are $\tilde{O}(n^{3/4} + m^{3/4})$. Note that by rebuilding the \conequery data structure earlier (or later), we can increase (or decrease) the update time and decrease (or increase) the query time, respectively. In particular, rebuilding the data structure once $k$ becomes larger than $K$, results in an amortized update time of $\tilde{O}((n + m)/K)$ and an expected query time of $\tilde{O}(K \sqrt{n})$.

\begin{restatable}{theorem}{thmRunningTime}
    We can build an insertion-only data structure for a dynamic simple polygon~$P$ with $m$ vertices and a dynamic set $S$ of $n$ sites in $P$ using $O(n (\log\log n + \log m) + m)$ space with $O(n^{3/4} \log n (\log \log n + \log^2m) + m^{3/4} \log \log n)$ amortized update time that can answer nearest neighbor queries in $O(n^{3/4}\log^{3/2} m + m^{1/4}(n^{1/2}\log^{3/2}m + \log n \log^2m))$ expected~time.
\end{restatable}

\begin{proof}
    Rebuilding the \conequery data structure once $k$ becomes larger than $\nOld^{1/4}+\mOld^{1/4}$ ensures that $k = O(n^{1/4} + m^{1/4})$. Theorems~\ref{thm:dynamic_sp} and~\ref{thm:conequery} then directly imply the query time.

    To obtain the amortized update time, we distribute the running time for constructing the \conequery data structure over the time steps up to the next update. Theorem~\ref{thm:conequery} states that the construction time is $O(\mOld \log \log \nOld + 
\nOld \log \nOld \log \log \nOld + 
\nOld \log \nOld \log^2 \mOld))$. There are $\nOld^{1/4}+\mOld^{1/4}$ time steps until the next rebuild, so the amortized cost of the rebuilding is:
\begin{align*}
   & \frac{O(\mOld \log \log \nOld + 
\nOld \log \nOld \log \log \nOld + 
\nOld \log \nOld \log^2 \mOld))}{\nOld^{1/4}+\mOld^{1/4}} \\
& = O\left(\frac{\mOld \log \log \nOld}{\nOld^{1/4}+\mOld^{1/4}} + \frac{\nOld \log \nOld \log \log \nOld}{\nOld^{1/4}+\mOld^{1/4}} + \frac{\nOld \log \nOld \log^2 \mOld}{\nOld^{1/4}+\mOld^{1/4}} \right)\\
& \leq O\left(\mOld^{3/4} \log \log \nOld + \nOld^{3/4} \log \nOld \log \log \nOld +\nOld^{3/4} \log \nOld \log^2 \mOld \right).
\end{align*}
This subsumes the update time of the \dynamicsp data structure.
\end{proof}

\section{Dynamic shortest path data structure}
\label{sec:dynamic_shortest_path}

The \dynamicsp data structure is an extension of the dynamic data
structure of Goodrich and Tamassia~\cite{dynamic_ray_shooting_sp},
discussed in Section~\ref{sec:preliminaries}.  We first explain
how we adapt it to support efficient access to our newly inserted
vertices. Then, we show how to use the \dynamicsp data structure to
generate $O(k)$ bounded nearest neighbor queries that can together
answer a nearest neighbor query, thus proving Theorem~\ref{thm:dynamic_sp}.

\subparagraph{Extending the dynamic shortest path data structure.}  We
extend the dynamic shortest path data structure of Goodrich and
Tamassia~\cite{dynamic_ray_shooting_sp} so that we have efficient
access to the vertices and edges that were inserted since the last
rebuild. We call these \emph{marked} vertices and edges.  In the data
structure of ~\cite{dynamic_ray_shooting_sp}, each polygonal chain is
stored in a \emph{chain tree}: a balanced binary tree where leaves
correspond to the edges in the chain and internal nodes to the
vertices of the chain. Each internal node also corresponds to a
subchain and stores its length. We additionally store a boolean in each node that that is true whenever its
subchain contains a marked edge, and false otherwise. 
Next to
this dynamic data structure, we keep a list of all marked
vertices. For this list, we consider the endpoint of an inserted edge
to be marked, even if the vertex already existed in $\POld$.

\begin{restatable}{lemma}{maintainingDS}
    The \dynamicsp data structure maintains a dynamic simple polygon with $O(\log^2m)$ update time using $O(m)$ space.
\end{restatable}

\begin{proof}
    The original dynamic shortest path data structure has $O(\log^2m)$ update time and $O(m)$ space~\cite{dynamic_ray_shooting_sp}. We additionally store a single boolean for each node, thus the space remains $O(m)$. These booleans can be updated in constant time per rotation, so splitting and splicing these trees remains $O(\log m)$ time. It follws that the $O(\log^2m)$ update time is retained.
\end{proof}

\begin{restatable}{lemma}{findingMarked}\label{lem:finding_marked}
    Given two points $p,q \in P$, the \dynamicsp data structure can find the first (or last) marked vertex on $\pi(p,q)$ in $O(\log^2m)$ time.
\end{restatable}

\begin{proof}
    To perform a shortest path query between points $p,q \in P$, Goodrich and Tamassia~\cite{dynamic_ray_shooting_sp} show that we can update the data structure such that $\pi(p,q)$ is one of the diagonals of the geodesic triangulation in $O(\log^2m)$ time. In particular, $\pi(p,q)$ is equal to one of the sides of the root geodesic triangle. This path is stored in at most two chain trees: one for the subpath that is the deltoid boundary, and one for the subpath that is the tail of the geodesic triangle. Note that the root geodesic triangle has one side that corresponds to an edge of $P$ and thus has at most one tail. To find the first (or last) vertex on $\pi(p,q)$, we perform the same operations on our $\dynamicsp$ data structure. We then search in both chain trees for the first (or last) marked vertex in $O(\log m)$ time.
\end{proof}

\subparagraph{Computing the bounded nearest neighbor queries.}
Next, we explain how to use the \dynamicsp data structure to find for
a query point $q \in P$ a set of bounded nearest neighbor queries such
that the nearest neighbor of $q$ is the result of one of these
queries. For each bounded nearest neighbor query we return the apex $v_R$, the two other endpoints $s,t$ of the shortest paths bounding $R$, and a boolean indicating whether $R$ is a geodesic triangle or geodesic cone. In case $R$ is a geodesic cone, we order $s,t$ such that the part of $\partial P$ in $R$ is given by the clockwise traversal from $s$ to $t$ along $\partial P$. We will choose our region $R$ in such way that the paths bounding $R$ are shortest paths in both $P$ and $\POld$. In other words, there are no marked vertices (that are reflex vertices) on these shortest paths.

\begin{lemma}
    We can compute a set of $O(k)$ bounded nearest neighbor queries in $\POld$ that allows us to find the nearest neighbor of a point $q \in P$ among the sites in $\SOld$ in $O( k \log^2 m)$~time.
\end{lemma}
\begin{proof}
    We first describe the partition of $P$ that will result in the $O(k)$ bounded nearest neighbor queries, and then discuss how to compute them. We partition $P$ into $O(k)$ maximal regions such that for every point in a region $R$, the path to $q$ uses the same \emph{marked} vertices, see Figure~\ref{fig:computing_regions}. This can be viewed as a shortest path map for $q$, where we only consider marked vertices. Each region in this map will be a bounded nearest neighbor query together with its apex, i.e. the first marked vertex on the path from any point in the region to $q$. 

    \begin{figure}
    \centering
    \includegraphics{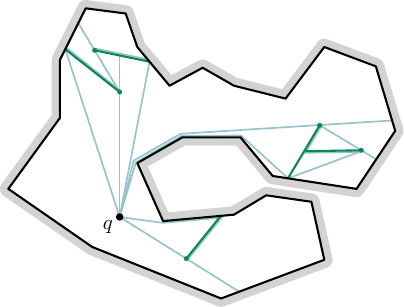}
    \caption{The bounded nearest neighbor query regions that are obtained for the query point $q$. The green segments are the segments inserted after the last rebuild.}
    \label{fig:computing_regions}
\end{figure}

    To find these regions, we keep track of all points on $\partial P$ that define the boundary between two such regions in a balanced binary tree $T$ based on the order along $\partial P$. 
    Observe that all marked vertices define such a boundary point. 
    We thus insert all $O(k)$ marked vertices into $T$.
    Then, for each marked vertex $v$, we:

    \newpage
    \begin{enumerate}
        \item Compute the shortest path $\pi(v,q)$.
        \item Perform a ray shooting query in the direction of the last segment of $\pi(v,q)$ to find the point $x$ where the extension of this segment hits $\partial P$.
        \item Insert $x$ into $T$. 
    \end{enumerate}
    
The \dynamicsp data structure supports shortest path queries in
$O(\log^2m)$ time. We can also obtain
the last segment on the shortest path using the chain trees representing the
path. The ray shooting query that follows again takes $O(\log^2m)$
time. Finally, inserting $x$ into $T$ takes just $O(\log k)$ time. As
there are $O(k) = O(m)$ marked vertices, this takes $O(k \log^2m)$ time in total.

What remains is to find the bounded nearest neighbor queries using $T$. Each pair $s,t$ of consecutive points in $T$ defines such a query. If $s$ and $t$ are endpoints of the same marked edge, then their query region $R$ is a geodesic triangle, otherwise $R$ is a geodesic cone. The apex $v_R$ of the bounded nearest neighbor is the last marked vertex on $\pi(q,s) \cap \pi(q,t)$, or $q$ if there is no such vertex. We distinguish the following cases:
\begin{itemize}
    \item If $s$ (or $t$) is a marked, then $v_R$ is the first (other) marked vertex on $\pi(s,q)$ (or~$\pi(t,q)$). 
    \item If both $s$ and $t$ are not marked, and thus included in $T$ in the ray shooting step, then $v_R$ is the second marked vertex on $\pi(s,q)$, i.e. the first vertex after the vertex used in the ray shooting step.
\end{itemize}
Note that in both cases $v_R$ is also on $\pi(t,q)$, as otherwise the ray shooting step on $v_R$ would add some point $x$ between $s$ and $t$ on $\partial P$ to $T$. Lemma~\ref{lem:finding_marked} implies that we can find $v_R$ in $O(\log^2m)$ time (to find the second marked vertex we can search the corresponding chain tree from the first marked vertex in $O(\log m)$ time).
\end{proof}

\section{Static cone query nearest neighbor data structure}
\label{sec:conequery}

\renewcommand{\mOld}{\ensuremath{m}}
\renewcommand{\POld}{\ensuremath{P}}
\renewcommand{\SOld}{\ensuremath{S}}
\renewcommand{\nOld}{\ensuremath{n}}

In this section, we present the static \conequery data structure. Throughout this section, let $P$ be a \emph{static} simple polygon with $m$ vertices and $S$ a \emph{static} set of $n$ sites in $P$. We want our data structure to answer bounded nearest neighbor queries
for a region $R$ and a query point $q$ on $\partial R$, which in our case will be the apex point $v_R$, as in
Definition~\ref{def:bounded_nn}. We assume that the region $R$ is given as endpoints of the (at most three) shortest paths bounding it, and in case $R$ is a geodesic cone, an indication which part of $\partial \POld$ is in $R$. 

The base of our data structure is a geodesic triangulation $\T$. We
construct several data structures for the nodes and edges of $\T$. We
assign each site $s \in S$ to the highest node in $\T$ whose geodesic
triangle contains $s$. Note that there is only one  node whose geodesic triangle contains $s$ if $s$
lies in some deltoid region, but there may be more if $s$ is on the
boundary of a geodesic triangle.
For a geodesic triangle $\delta \in \T$, we denote the set of sites assigned to $\delta$ by~$S_\delta$. For the deltoid $\Delta$ of $\delta$, we define $S_\Delta := S_\delta \cap \Delta$.

As a preprocessing step, we triangulate $P$~\cite{Chazelle_triangulate} and construct the shortest path data structure
of Guibas and Hershberger~\cite{2PSP_simple_polygon}. We also construct a point location data structure~\cite{Kirkpatrick_point_location} on the geodesic triangulation. Each face in this subdivision is associated with its node in $\T$, and each edge is associated with the highest node in $\T$ that has this edge on its boundary. The geodesic triangulation can also be used to answer ray shooting queries in $O(\log \mOld)$ time~\cite{Balanced_geodesic_triangulation}.

Next, we define two fundamental
queries for a node in $\T$.

\begin{definition}\label{def:type_a_b}
    Let $\delta$ be a node in $\T$. Let $(R,q)$ be a bounded nearest neighbor query, then we define the following subqueries.
    \begin{itemize}
        \item \emph{Type (a)}: return the nearest neighbor of $q$ in $S_\delta \cap R$.
        \item \emph{Type (b)}: let $\T_x$ be the subtree of $\T$ rooted at a child $x$ of $\delta$ and $P_x$ the corresponding subpolygon, i.e. $P_x := \bigcup_{y \in \T_x} y$. For $q \notin P_x$, return the nearest neighbor of $q$ in $S \cap P_x$.
    \end{itemize}
\end{definition}

For each node $\delta \in \T$, we store two data structures that together can answer type (a) queries. One that finds the nearest neighbor of $q$ among the sites on the boundary of the geodesic triangle, and one that finds the nearest neighbor of $q$ among the sites in the deltoid region. For each edge $(\delta, x) \in \T$ we store a data structure that can answer type (b) queries. 

Next, we show that we can answer bounded nearest neighbor queries using $O(\log \mOld)$ type~(a) and (b) queries. We first state two lemmas that prove useful in finding these queries.

\begin{restatable}{lemma}{sideTest}
  \label{lem:side_test_queries}
  Let $p,r$ be two points on $\partial \POld$, and let $s \in
  \POld$. In $O(\log m)$ time we can test if $s$ lies right of the
  shortest path $\pi(p,r)$.
\end{restatable}

\begin{proof}
  Let $v_1,..,v_m$ denote the vertices of $\POld$ ordered along the
  boundary of $\POld$. We query the Guibas and Hershberger data
  structure~\cite{2PSP_simple_polygon} for the shortest path $\pi(p,r)$. This data
  structure can actually report a sequence of $O(\log m)$ hourglasses
  $H_1,..,H_{O(\log m)}$ so that $\pi(p,r)$ is the path that we obtain
  by concatenating these hourglasses. An \emph{hourglass} $H_i$ is
  defined by two diagonals $d_i$ and $d_{i+1}$ of $\POld$ and is the
  union of all shortest paths between pairs of points in
  $d_i \times d_{i+1}$. These $O(\log m)$ diagonals actually
  partition the polygon $\POld$ into $O(\log n)$ subpolygons
  $P'_1,..,P'_{O(\log m)}$.

  We locate the triangle in the triangulation of $\POld$ used
  by the shortest path data structure, defined by vertices $v_a$,
  $v_b$, and $v_c$, containing $s$. Based on the indices $a$, $b$, and
  $c$ we can now test which subpolygon $P'_i$ contains the triangle and
  thus $s$. In particular, by comparing the indices to the indices of
  the endpoints of the diagonals $d_i$. This takes $O(\log m)$ time.

  We now concatenate the first $i$ hourglasses into a funnel $F$
  (essentially a degenerate hourglass) to diagonal $d_{i+1}$, and test
  if $s$ lies left of $F$, inside $F$, or right of $F$.  We can do
  this in $O(\log m)$ time by a binary search on the vertices of
  $F$.

  If $s$ lies outside of $F$, we already have our answer and we can
  return. If $s$ lies inside $F$, we iteratively concatenate the
  remaining hourglasses $H_{i+1},..,H_{O(\log m)}$ onto the
  funnel. Every concatenation step involves the computation of two
  inner tangents. By computing on which side of these tangents $s$ is,
  we can maintain whether $s$ lies left, inside, or right, of the
  funnel.

  In total, the time required to concatenate all hourglasses is just
  $O(\log m)$ (as this is just the query algorithm to compute
  $\pi(p,r)$). Testing on which sides of each of the inner tangents
  $s$ lies is just a constant amount of extra work per concatenation
  operation. Hence, the total query time remains $O(\log m)$.
\end{proof}

\begin{restatable}{lemma}{intersectionSP}\label{lem:intersection_SP}
    Let $p,r$ be two points on $\partial \POld$, and let $s,t \in \POld$. We can determine in $O(\log^2 m)$ time the first and last point of $\pi(s,t) \cap \pi(p,r)$, or conclude the paths are disjoint.
\end{restatable}

\begin{proof}
    Let $a$ and $b$ be the first and last intersection point of the two shortest paths, if they exist, as in Figure~\ref{fig:touching_boundary}. Note that two shortest paths intersect in at most one connected subpath, because otherwise either path could be shortcut via the other path.
    We will use Lemma~\ref{lem:side_test_queries} as a subroutine to find $a$ and $b$. This lemma states that, given the Guibas and Hershberger data structure~\cite{2PSP_simple_polygon}, we can test in $O(\log m)$ time whether a point lies to the right of a shortest path between two points on $\partial P$. We first apply Lemma~\ref{lem:side_test_queries} to test whether $s$ and $t$ lie on the same or different sides of $\pi(p,r)$. In case one or both of $s$ and $t$ lies on $\pi(p,r)$, we consider this as them being on different sides. Next, we discuss how to handle these two cases.
\begin{figure}
    \centering
    \includegraphics[page=3]{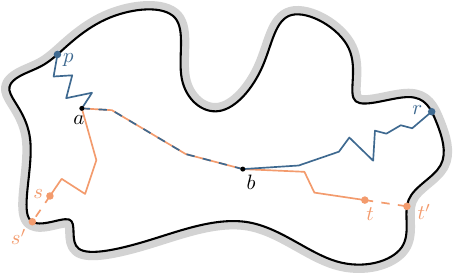}
    \caption{The shortest paths $\pi(p,r)$ and $\pi(s,t)$ overlap between $a$ and $b$.}
    \label{fig:touching_boundary}
\end{figure}

    If $s$ and $t$ lie on different sides of (or on) $\pi(p,r)$, we assume without loss of generality that $s$ is right of (or on) $\pi(p,r)$. We perform a binary search on the vertices of $\pi(s,t)$ to find the first vertex $v$ that is no longer to the right of $\pi(p,r)$ using Lemma~\ref{lem:side_test_queries} as a subroutine. A shortest path query to the  Guibas and Hershberger data structure essentially gives us a chain tree representing $\pi(s,t)$ that we can use to access a vertex in $O(\log m)$ time.
    Note that every vertex of $\pi(s,t)$ after $v$ must be on or left of $\pi(p,r)$ as well, as shortest paths can intersect at most once. To find $a$, we also need to find the edge of $\pi(p,r)$ that is intersected by $\pi(s,t)$. We extend the first and last segment of $\pi(s,t)$ to points $s'$ and $t'$ the boundary of $\POld$ (as illustrated in Figure~\ref{fig:touching_boundary}) using a ray shooting query. To find the edge of $\pi(p,r)$ that is intersected by $\pi(s,t)$, we now perform a binary search on $\pi(s',t')$ as before. Using the edges of $\pi(s,t)$ and $\pi(p,r)$ that intersect, we can compute the location of $a$ in constant time. We can find $b$ in a symmetric way. Note that $b$ only exists if $a$ and $b$ are both vertices of $\pi(p,r)$, and thus the second binary search to find the exact intersection is obsolete.

    If $s$ and $t$ are on the same side of $\pi(p,r)$, then $\pi(s,t)$ might share a subpath with $\pi(p,r)$, as in Figure~\ref{fig:touching_boundary}. In this case we cannot perform a binary search on $\pi(s,t)$, as all vertices lie on the same side of (or on) $\pi(p,r)$. Instead, we binary search on the vertices of $\pi(s,r)$. We next prove the claim that $\pi(s,r) = \pi(s,a) \cup \pi(a,r)$. 
    
    Once $\pi(s,r)$ intersects $\pi(p,r)$, it must follow this shortest path up to $r$, as it cannot intersect the path twice. The path $\pi(s,r)$ thus cannot intersect $\pi(p,r)$ before $a$, as this would imply the subpath $\pi(s,a)$ is not a shortest path. Now, suppose that $\pi(s,r)$ would not intersect $a$, i.e. it would join $\pi(p,r)$ at a later vertex. Consider the point $t'$ on $\partial \POld$ that is hit when extending the last segment of $\pi(s,t)$. Then $\pi(s,r)$ must intersect $\pi(s,t')$ somewhere to reach $r$. However, as both $\pi(s,t')$ and $\pi(s,r)$ start at $s$, this means that these shortest paths must coincide up to their intersection point, and thus $a$ is on $\pi(s,r)$. The claim follows.
    
    As each vertex of $\pi(s,r)$ after $a$ is on $\pi(a,r)$, we can perform a binary search on $\pi(s,r)$ to find $a$ by applying Lemma~\ref{lem:side_test_queries} in each step. Similarly, we can find $b$ using a binary search on $\pi(t,p)$. Finally, we need to check whether $a$ and $b$ indeed lie on $\pi(s,t)$, as we might find false positives in case $\pi(s,t) \cap \pi(p,r) = \emptyset$. For example, by applying Lemma~\ref{lem:side_test_queries} again to $\pi(s',t')$ and $a$, where $s'$ is obtained by extending the first segment of $\pi(s,t)$ up to $\partial \POld$.

    In both cases, we perform a binary search on a shortest path where each check takes $O(\log m)$ time. The running time is thus $O(\log^2m)$.
\end{proof}

\begin{lemma}\label{lem:number_of_type_ab_queries}
    A bounded nearest neighbor query can be answered using $O(\log \mOld)$ type (a) and~(b) queries. We can determine these $O(\log \mOld)$ queries in $O(\log^2 \mOld)$ time.
\end{lemma}
\begin{proof}
We first prove the first part of the lemma, and then show how to find
these type (a) and (b) queries
efficiently. Figure~\ref{fig:type_b_queries} illustrates which nodes
and edges we query. Let $(R,q)$ be a bounded nearest neighbor
query. The region $R$ is either bounded by three shortest paths, or by
two shortest path and part of $\partial \POld$. We first consider one
of these shortest path that bounds $R$. Let $\pi(s,t)$ be this
shortest path, and let $v(s)$ and $v(t)$ be geodesic triangles in $\T$
that contain $s$ and $t$ respectively. Let $\sigma$ be the subset of
nodes of $\T$ on the paths from $v(s)$ and $v(t)$ to the root of
$\T$. Furthermore, let $P_\sigma$ be the union of the geodesic triangles in~$\sigma$, i.e. $P_\sigma = \bigcup_{\delta \in \sigma} \delta$. We claim that $\pi(s,t)$ is contained $\POld_\sigma$. 
Indeed, consider a deltoid region $\Delta$ that is entered by $\pi(s,t)$ of a
geodesic triangle $\delta \notin \sigma$. Let $\pi(a,b)$ be the shortest path that is the boundary of $\delta$ through which $\pi(s,t)$ enters $\Delta$. Note that by definition $a$ and~$b$ are on~$\partial \POld$. So, because $\pi(s,t)$ starts and
ends in $\POld_\sigma$,
it should at some point reenter $\POld_\sigma$ by crossing
$\pi(a,b)$ again. This contradicts $\pi(s,t)$ being a shortest path.

\begin{figure}
    \centering
    \includegraphics{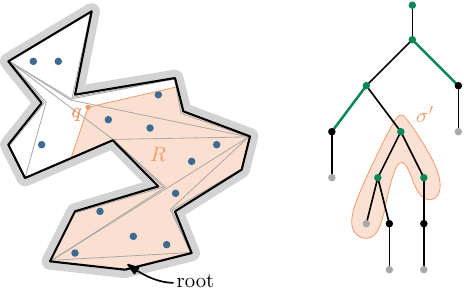}
    \caption{A geodesic triangulation and its dual tree. Nodes with empty deltoid regions are gray. We perform a type (a) query for all green nodes, and a type (b) query for all green edges.}
    \label{fig:type_b_queries}
\end{figure}

Let $\sigma'$ be the the union of the $\sigma$ sets for the shortest
paths bounding $R$, excluding any geodesic triangles that intersect
none of the shortest paths bounding $R$. Observe that for any point $p \in P$ the
subgraph of nodes in $\T$ whose geodesic triangle contains $p$ is
connected. As each pair of shortest paths bounding $R$ has a common
endpoint, the subgraph of $\sigma'$ is also connected. For each
$\delta \in \sigma'$, we perform a type (a) query. Because $\T$ has
height $O(\log \mOld)$, we perform $O(\log \mOld)$ of these
queries. Together, these queries find the nearest neighbor of $q$
in~$R$ for all geodesic triangles whose intersection with one of the
shortest paths bounding~$R$ is non-empty. Clearly, any node whose
deltoid region is intersected by such a shortest path is in
$\sigma'$. Any site that is on the boundary of a geodesic triangle
intersected by such a shortest path is stored in the highest node in
$\T$ that contains the site. As this subset of nodes is connected,
some node in $\sigma'$ must contain this site.

What remains is to query the geodesic triangles that are fully contained in $R$. If $R$
is a geodesic triangle, no node in~$\T$ can be contained in the
interior of $R$, as there are no vertices of $\POld$ in its interior,
and the corners of the geodesic triangle are vertices of~$\POld$. If
$R$ is a geodesic cone, there might be $O(\mOld)$
nodes in $\T$ contained inside $R$ that do not
intersect the two bounding paths.  Because~$R$ is connected,
there is an edge $(\delta,x)$ in $\T$, with $\delta$ the parent of $x$,
whose removal splits $\T$ into two subtrees, one of
which contains exactly the geodesic triangles that are fully contained
in~$R$. If this subtree is $\T_x$ (recall that $\T_x$ is the subtree rooted at $x$), then we perform a
type (b) query on $(\delta,x)$. If the subtree is $\T \setminus \T_x$, then for each node
on the path from $\delta$ to the root, we perform a type (a) query and
a type (b) query for the incident edge not on this path.

Finally, we prove the second part of the lemma. We first compute the
set $\sigma'$. Consider one of the shortest paths $\pi(s,t)$ that
bounds $R$. We first perform a point location query for both $s$ and
$t$. This gives us two nodes $v(s),v(t)$ in $\T$ that contain $s$ and
$t$. We then find the lowest common ancestor $w$ of $v(s)$ and $v(t)$
in $O(\log \mOld)$ time. All nodes on the paths from $v(s)$ and $v(t)$
to $w$ are in $\sigma'$. Note that the previous argument that
$\pi(s,t)$ cannot cross a geodesic triangle boundary and return even
implies that $\pi(s,t)$ is contained in the subpolygon correspoding to
these nodes. However, there may be more nodes on the path from $w$ to
the root of $\T$ whose boundary intersects $\pi(s,t)$. Let $\pi(p,r)$ be
the shortest path that is the boundary between $w$ and its parent
geodesic triangle. 
Lemma~\ref{lem:intersection_SP} implies that we can
find the first point $a$ and the last point $b$ of $\pi(s,t) \cap \pi(p,r)$, or conclude the paths are disjoint, in $O(\log^2m)$
time.
Any geodesic triangle in $\T$ above $w$ that intersects $\pi(s,t)$ intersects either $a$ or $b$ (or both).
It follows we can find the highest node of $\T$ that intersects $\pi(s,t)$ by a point location query on $a$ and $b$, and choosing the highest node of the returned nodes. All nodes on the path in $\T$ from $w$ and to this node are in $\sigma'$. By following the same procedure for all shortest path bounding $R$, we can compute $\sigma'$ in $O(\log^2 m)$ time.

If $R$ is a geodesic triangle, we are finished. If $R$ is instead a
geodesic cone, we need to find the edge whose removal splits off the
fully contained geodesic triangles as described earlier. Note that
this edge must be incident to $\sigma'$. Let $s$ and $t$ be the
endpoints other than $q$ of the shortest paths that bound $R$. We
first find the subsequence of polygon vertices contained in~$R$ by a
point location query for $s$ and $t$. By storing the boundary of $\POld$
in a balanced tree, we can then find the subsequence of vertices
between the edges containing $s$ and $t$ in $O(\log \mOld)$ time. For each node in
$\sigma'$, we check whether the incident edge that is not in
$\sigma'$ is contained in $R$ by checking whether its endpoints are
contained in the subsequence in $O(\log \mOld)$ time. We thus find the
desired edge of $\T$ in $O(\log^2 \mOld)$ time. The required type (a)
and (b) queries can then be found by traversing the tree as described
before in $O(\log \mOld)$ time.
\end{proof}

\subsection{Second level data structures}

For each node $\delta \in \T$, we store two data structures that together can answer type~(a) queries, and for each edge $(\delta,x)$ we store a data structure that can answer type~(b) queries.

\subparagraph{Data structures for type (a) queries.}
Fix a node $\delta \in \T$ and let $\Delta$ be its deltoid. For each shortest path bounding $\delta$, we store the sites of $S_\delta$ that are on the path in a binary search tree based on the order along the path. Together, these trees contain the sites in $S_\delta \setminus S_\Delta$.

On the sites in $S_\Delta$ we construct Chan's partition tree~\cite{Chan_partition_tree}. Each leaf node stores a constant number of sites, and each site is stored in exactly one leaf. 
For each node $\mu$ in the partition tree, let $S(\mu)$ denote the sites stored in the leaves below $\mu$ and let $\tri(\mu)$ denote the triangular cell of $\mu$ that contains $S(\mu)$. We store the following data structures for $\mu$, see Figure~\ref{fig:typeA}:
\newpage
\begin{enumerate}
    \item The geodesic Voronoi diagram of the sites $S(\mu)$ with
      respect to the polygon $\tri(\mu) \cap \Delta$.
    \item For each segment bounding $\tri(\mu)$: the forest
      representing the Voronoi diagram of $S(\mu)$ restricted to the
      subpolygon bounded by $\partial \POld$ and the chord that
      contains the segment that does not contain
      $S(\mu)$.
    \item For each shortest path $\pi$ bounding $\delta$: the forest representing the Voronoi diagram of $S(\mu)$ restricted to the subpolygon bounded by $\partial \POld$ and $\pi$ that does not contain $S(\mu)$.
\end{enumerate}
\begin{figure}
    \centering
    \includegraphics{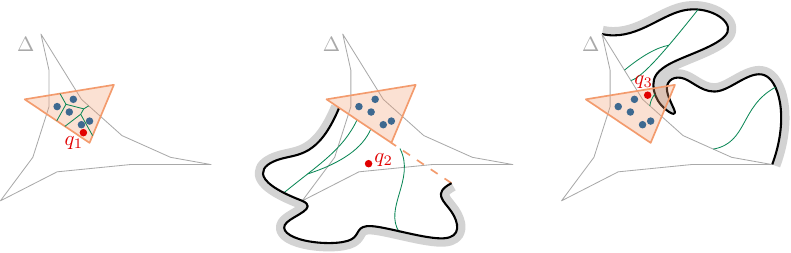}
    \caption{The three data structures stored for the orange cell $\tri(\mu)$ for answering type (a) queries. The first is used for a query point inside $\tri(\mu)\cap \Delta$, the second for a query point outside $\tri(\mu)$, and the third for a query point that is in $\tri(\mu)$ but outside of $\Delta$.}
    \label{fig:typeA}
\end{figure}

The first data structure is simply the standard geodesic Voronoi diagram~\cite{Geodesic_voronoi_diagram}.
For the second and third data structure, we adapt a data structure that Agarwal, Arge, and Staals~\cite{dynamic_geodesic_nn} use in their dynamic nearest neighbor data structure. They essentially prove that when a polygon is partitioned into two subpolygons $P_\ell, P_r$ by a line segment, the forest representing the Voronoi diagram of the sites in $P_\ell$ restricted to $P_r$ can be stored and queried efficiently. In Section~\ref{sec:forest_ds_shortest_path}, we generalize this result to the case where the polygon is partitioned using an arbitrary shortest path instead of a line segment, and obtain the following result. 

\begin{restatable}{theorem}{forestBySP}\label{thm:forest_by_sp}
    Let $P$ be a polygon with $m$ vertices appropriately preprocessed for two-point shortest path and ray-shooting queries. Let $P_\ell, P_r$ be a partition of $P$ into two subpolygons by a shortest path between two points on $\partial P$. Given a set of $n$ sites $S$ in $P_\ell$, there is an $O(n)$ size data structure storing the combinatorial
  representation of the Voronoi diagram of $S$ in $P_r$ so that given a query point $q \in P_r$, we can find the site in $S$ closest to $q$ in $O(\log n \log m)$
  time. Building the data structure takes $O(n \log^2 m + n \log n)$ time.
\end{restatable}

\subparagraph{Data structure for type (b) queries.}
Let $\delta$ be a node in $\T$ and $x$ be a child of $\delta$. To be
able to answer type (b) queries we store Voronoi diagram of the sites
in $S \cap P_x$ restricted to $P \setminus P_x$, for $P_x$ as in Definition~\ref{def:type_a_b}, in the forest data structure of Theorem~\ref{thm:forest_by_sp}, as the boundary between the two geodesic triangles is a shortest path.

\subsection{Analysis}
Next, we prove that the \conequery data structure achieves the bounds in Theorem~\ref{thm:conequery}. We use the following lemma about multilevel data structures using Chan's partition tree.
\begin{lemma}[Corollary 6.2 of~\cite{Chan_partition_tree}]\label{lem:canonical_subsets}
    Given $n$ points in $\R^2$, for any fixed $\gamma < 1/2$, we can form $O(n)$ canonical subsets of total size $O(n \log \log n)$ in $O(n \log n)$ time, such that the subset of all points inside any query triangle can be reported as a union of disjoint canonical subsets $C_i$ with $\sum_i |C_i|^\gamma \leq O(\sqrt{n})$ in time $O(\sqrt{n})$ w.h.p.$(n)$. 
\end{lemma}

\subparagraph{Query time.}
Let $(R,q)$ be a bounded nearest neighbor query. Lemma~\ref{lem:number_of_type_ab_queries} states that we can answer the query using $O(\log \mOld)$ type (a) and (b) queries and that we can compute these queries in $O(\log^2 \mOld)$ time. What remain is to analyze the time of a type (a) and (b) query.

First, consider a type (a) query. Recall that this asks for the
nearest neighbor of $q$ in $S_\delta \cap R$ for a geodesic triangle
$\delta \in \T$ with deltoid $\Delta$. To answer this query we have to
query the three binary trees storing the sites in $S_\delta \setminus
S_\Delta$ and the partition tree on $S_\Delta$. For a shortest path of
$\partial \delta$, the sites on this path that are contained in $R$
form a contiguous subsequence. We can thus obtain this subset of sites
in a binary tree in $O(\log n)$ time. Because the distance from $q$ to
a shortest path is a convex function, we can then binary search in this subtree to find the site closest to $q$ in $O(\log n \log m)$ time.

To find the closest site of $q$ in $S_\Delta$, we query the partition tree. However, we cannot apply Lemma~\ref{lem:canonical_subsets} directly, as this requires the query region to be a triangle instead of a geodesic triangle or cone. Oh and Ahn~\cite{dynamic_convex_hull} essentially prove that seven triangles can cover a geodesic triangle without including any sites of $S_\Delta$ outside of the geodesic triangle in their interior. Here, we need a slightly more general statement that also includes geodesic cones.

\begin{restatable}{lemma}{coveringWithTriangles}\label{lem:constant_number_triangles}
    We can construct in $O(\log^2 \mOld)$ time a constant number of interior disjoint triangles whose union contains $\Delta \cap R$ but does not contain any sites in $S_\Delta \setminus R$.
\end{restatable}

\begin{proof}
  There are no vertices of $\POld$ interior to $\Delta$, as each convex chain bounding $\Delta$ is a shortest path in $\POld$. It follows that any vertex of $\Delta \cap R$ must be a vertex of $\Delta$, or the intersection point of a shortest path bounding $R$ and $\partial \Delta$.
  The polygon $\Delta \cap R$ thus consists of at most three convex chains
  and six line segments. 
  If~$R$ is a geodesic triangle, these convex
  chains consist of edges of the boundary of both $\Delta$
  and~$R$. If~$R$ is a geodesic cone, then these edges are still
  boundary edges of $\Delta$, but they might not be edges of $R$. We
  reduce the complexity of this region by replacing each such convex
  chain by the line segment connecting the endpoints of the chain, see
  Figure~\ref{fig:geodesic_to_euclidean}. As the area bounded by this
  line segment and the convex chain is disjoint from $\Delta$, no
  additional sites of $S_\Delta$ are included in this polygonal
  region. The new region is bounded by nine line segments, and thus it
  can be triangulated using seven Euclidean triangles. Finding the
  vertices of $\Delta \cap R$ can be done in $O(\log^2m)$ time using
  the tree representations of the shortest paths obtained by querying the Guibas and Herschberger data structure~\cite{2PSP_simple_polygon}.
    \begin{figure}
        \centering
        \includegraphics[page=2]{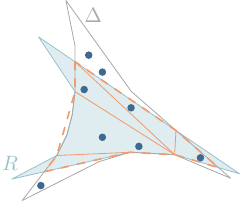}
        \caption{The five orange Euclidean triangles cover $\Delta \cap R$ without including any sites of $S_\Delta$ that lie outside of $R$.}
        \label{fig:geodesic_to_euclidean}
    \end{figure}
\end{proof}

For each of the $O(1)$ triangles given by Lemma~\ref{lem:constant_number_triangles} we query the partition tree. Lemma~\ref{lem:canonical_subsets} states that the partition tree gives us a set of disjoint canonical subsets $C_i$.
For each canonical subset, which corresponds to some node $\mu$ in the partition tree, we perform a nearest neighbor query using one of the data structures stored in the node. Which data structure we query depends on the location of $q$ in relation to $\Delta$ and $\tri(\mu)$, see Figure~\ref{fig:typeA}.

\begin{itemize}
    \item If $q$ is inside both the deltoid $\Delta$ and the cell $\tri(\mu)$, i.e. $q \in \Delta \cap \tri(\mu)$, we query the geodesic Voronoi diagram of $\Delta \cap \tri(\mu)$.
    \item If $q$ is outside the cell $\tri(\mu)$, i.e. $q \notin \tri(\mu)$, we query the forest constructed for the segment of $\tri(\mu)$ whose extension separates $q$ from $\tri(\mu)$.
    \item If $q$ is inside the cell $\tri(\mu)$ but outside the deltoid $\Delta$, i.e. $q \in \tri(\mu)$ and $q \notin \Delta$, we query the
      forest constructed for side of $\delta$ that separates $q$ from
      $\Delta$.
\end{itemize}

Consider a canonical subset $C_i$ for which the complexity of $\Delta$ restricted to the cell of $C_i$ is $m_i$. The first type of query, which is in the (regular) geodesic Voronoi diagram, takes $O(\log(|C_i|+m_i))$ time~\cite{Geodesic_voronoi_diagram}. The second and third type of query take $O(\log |C_i| \log \mOld)$ time by  Theorem~\ref{thm:forest_by_sp}. 
The query time thus becomes 
\begin{align*}
    \sum_i (\log |C_i| \log \mOld + \log ( |C_i| + m_i)) &= \log \mOld \sum_i \log |C_i| + \sum_i \log(|C_i| + m_i)\\
    &\leq \log \mOld \sum_i \log |C_i| + \sum_i (\log|C_i| + \log m_i)\\
    &= (\log \mOld + 1) \sum_i \log |C_i| + \sum_i \log \mOld.
\end{align*}
Here we used that for $|C_i|, m_i \geq 2$ it holds that $\log(|C_i| + m_i) \leq \log (|C_i| m_i)$. Choosing any $\gamma \in (0, 1/2)$ in Lemma~\ref{lem:canonical_subsets} implies that $\sum_i \log |C_i| = O(\sqrt{|S_\Delta|})$. The number of terms in the sum $\sum_i \log \mOld$ is bounded by $O(\sqrt{|S_\Delta|})$ as well. A type (a) query can thus be answered in $O(\sqrt{|S_\Delta|} \log \mOld)$ time.

Finally, we note that the $O(\log \mOld)$ type (a) queries we perform
are on distinct deltoids. Thus, if we sum over all deltoids $\Delta_i$
that we query, we have $\sum_i |S_{\Delta_i}| = O(n)$. To bound the sum of square roots over these set sizes, we use the HM-GM-AM-QM inequalities, which states that for positive reals $x_1,\dots,x_a$, we have $x_1 + \dots + x_a \leq \sqrt{a (x_1^2 + \dots + x_a^2)}$. It then follows that the total query
time of all type (a) queries is $\sum_{i = 0}^{O(\log \mOld)}
\sqrt{|S_{\Delta_i}|} \log \mOld = O(\sqrt{n} \log^{3/2} \mOld)$.

A type (b) query requires only a single query of the data structure of Theorem~\ref{thm:forest_by_sp}, and thus takes $O(\log n \log \mOld)$ time. The total query time of the type (b) queries is $O( \log n \log^2 \mOld)$.

\subparagraph{Construction time.}
As preprocessing, we build the shortest path data structure of Guibas and Hershberger in $O(\mOld)$ time~\cite{Chazelle_triangulate, 2PSP_simple_polygon}. We then build the level one data structure: a balanced geodesic triangulation together with a point location data structure, in $O(\mOld)$ time~\cite{Balanced_geodesic_triangulation}.

There are three level two data structures: one for each edge of $\T$ (for the type (b) queries), and two for each node of $\T$ (for the type (a) queries). Let $\delta$ be a node in $\T$ and $\Delta$ the corresponding deltoid. 

First, consider the data structure on each edge of $\T$. This is the
forest data structure of
Theorem~\ref{thm:forest_by_sp}. For a child $x$ of $\delta$ this data structure can be build in
$O(|S \cap P_x| \log^2 \mOld + |S \cap P_x| \log |S \cap P_x|)$ time,
with $P_x$ as in Definition~\ref{def:type_a_b}. Note that each site in
$S$ can occur in at most $O(\log \mOld)$ polygons $P_x$, once for each
level of $\T$. It follows that constructing the data structure for
each edge of $\T$ takes $O(n \log^3 \mOld + n \log n \log \mOld)$ time
in total.

Second, consider the two data structures on the nodes of $\T$. For each shortest path bounding $\delta$, we store the sites on the path in a binary search tree. These trees can be constructed in $O(n_\delta (\log n_\delta +\log \mOld))$ time. Furthermore, we store the sites in the deltoid $\Delta$ in a partition tree. The partition tree (excluding auxiliary data structures) can be build in $O(n_\Delta\log n_\Delta)$ time (Lemma~\ref{lem:canonical_subsets}).

For each node $\mu$ of the partition tree, we store three third level Voronoi data structures for the sites $S(\mu)$. Let $m(\mu)$ be the number of vertices of $\tri(\mu) \cap \Delta$. The first level three data structure is a (regular) geodesic Voronoi diagram restricted to $\mu$, including a point location data structure, and can be constructed in $O(m(\mu) + |S(\mu)| \log |S(\mu)|)$ time~\cite{Geodesic_voronoi_diagram}. The second and third data structures are the forest data structure of Theorem~\ref{thm:forest_by_sp} on each segment bounding the cell of $\mu$ and each shortest path bounding $\delta$. These can be constructed in $O(|S(\mu)| \log^2 \mOld + |S(\mu)|\log|S(\mu)|)$ time (Theorem~\ref{thm:forest_by_sp}). The total construction time for a node $\mu$ of the partition tree is then $O(m(\mu) + |S(\mu)| \log |S(\mu)| + |S(\mu)| \log^2 \mOld)$. To bound the total construction time for all nodes of the partition tree, observe that the cells in each of the $O(\log \log n_\Delta)$ levels of the partition tree are disjoint. So, $\sum_\mu |m(\mu)| = O(m_\Delta \log \log n_\Delta)$. Furthermore, $\sum_\mu |S(\mu)| = O(n_\Delta \log \log n_\Delta)$. The total construction time over all nodes is then $O(m_\Delta \log \log n_\Delta + n_\Delta \log \log n_\Delta (\log n_\Delta + \log^2 \mOld))$.

Together, the level two and three data structures on a node $\delta \in \T$ are constructed in $O(m_\delta \log \log n_\delta + n_\delta \log \log n_\delta (\log n_\delta + \log^2 \mOld))$ time. This implies that the total construction time for all nodes in $\T$ is $O(\mOld \log \log n + n (\log n \log \log n + \log^2 \mOld \log \log n))$. 

The total construction time of all first, second, and third level data structures then becomes 
$O(\mOld \log \log n + 
n \log n \log \log n + 
n \log^3 \mOld +
n \log n \log^2 \mOld +
n \log n \log \mOld) = O(\mOld \log \log n + 
n \log n \log \log n + 
n \log n \log^2 \mOld)$.

\subparagraph{Space usage.}
The space of the shortest path data structure on $\POld$ is $O(\mOld)$.
For a node $\delta \in \T$ and $\Delta $ the corresponding deltoid, the trees for $\partial \delta$ have size $O(n_\delta)$. As for the partition tree, Lemma~\ref{lem:canonical_subsets} states that the total size of the canonical subsets is $O(n_\Delta \log \log n_\Delta)$. The three data structures that are built for each node of the partition tree use just linear space in the number of sites, thus the total size of the partition tree including its auxiliary data structures remains $O(n_\Delta \log \log n_\Delta)$. The total space over all nodes in $\T$ is thus $O(n \log \log n)$.
What remains is to analyze the space of the data structures for type (b) queries. Theorem~\ref{thm:forest_by_sp} states that this forest data structure for a child $x$ of $\delta$ uses $O(|S \cap P_x|)$ space. As before, each site in $S$ can occur in at most $O(\log \mOld)$ polygons $P_x$, once for each level of $\T$. Thus the total space of these data structures is $O(n \log \mOld)$.

\renewcommand{\POld}{\ensuremath{\tilde{P}}}
\renewcommand{\mOld}{\ensuremath{\tilde{m}}}
\renewcommand{\SOld}{\ensuremath{\tilde{S}}}
\renewcommand{\nOld}{\ensuremath{\tilde{n}}}

\section{Generalizing the geodesic Voronoi diagram data structure to
  shortest paths}
\label{sec:forest_ds_shortest_path}

\renewcommand{\mOld}{\ensuremath{m}}
\renewcommand{\POld}{\ensuremath{P}}
\renewcommand{\SOld}{\ensuremath{S}}
\renewcommand{\nOld}{\ensuremath{n}}
Throughout this section, let $\POld$ be a static simple polygon with
$\mOld$ vertices. We assume that $\POld$ is appropriately preprocessed
for two-point shortest path queries~\cite{2PSP_simple_polygon} and
ray-shooting queries~\cite{Balanced_geodesic_triangulation}. This
takes $O(\mOld)$ time and space. Let $\lambda$ be a shortest path that partitions $\POld$ into two
subpolygons, $P_\ell$ and $P_r$, and let $S$ be a set of $n$ sites in
$P_\ell$. Note that we redefine $S$ here to be limited to sites in $P_\ell$, as these are the only relevant sites for this section. Our goal is to construct a data structure that allows us to
find the closest site in $S$ for a query point $q \in P_r$. Agarwal,
Arge, and Staals~\cite{dynamic_geodesic_nn, dynamic_geodesic_nn_arxiv} show that when $\POld$ has
been appropriately stored for two-point shortest path queries, and
$\lambda$ is a single line segment, one can in
$O(n(\log n + \log^2 \mOld))$ time compute an implicit representation of
the Voronoi diagram $\Vor$ restricted to $P_r$. This implicit
representation uses $O(n)$ space and can answer nearest neighbor
queries in $O(\log^2 m)$ time. We will generalize their result to the
case where $\lambda$ is an arbitrary shortest path.

We will argue that: 1) $\Vor$ is a forest that has $O(n)$ degree one
or three vertices and $O(\mOld)$ degree two vertices, and 2) that it
suffices to compute only the degree one and three vertices and the
topology of the forest to answer point location queries. Like Agarwal,
Arge, and Staals~\cite{dynamic_geodesic_nn}, we show that $\Vor$ can
be regarded as an Hamiltonian abstract Voronoi
diagram~\cite{klein1994hamiltonian_vd} of an ordered subset
$S' \subseteq S$ of the sites. Once we establish this fact, we can use
their algorithm more or less as is. The main difference is in the
representation of $\Vor$ for point location queries that we
use. Agarwal, Arge, and Staals~\cite{dynamic_geodesic_nn} use that the
edges of $\Vor$ (that are not stored explicitly) are monotone with
respect to the direction perpendicular to~$\lambda$. However, when
$\lambda$ is not a segment, this is no longer well defined. Hence, we
provide a different approach instead.

\subsection{$\Vor$ as an Hamiltonian abstract Voronoi diagram}
\label{sub:Hamiltonian_avd}

We first establish the following technical lemma.

\begin{lemma}
  \label{lem:one_intersection}
  Let $s,t \in S$ be two sites in $P_\ell$. The bisector $b_{st}$ of
  $s$ and $t$ in $\POld$ intersects $\lambda$ in at most one point
  $w$. Moreover, we can compute this point in $O(\log^2 m)$ time.
\end{lemma}

\begin{proof}
  Aronov~\cite{aronov1989geodesic} already argued that the bisector of
  two sites intersects the boundary of the polygon in two
  points. Applying this result on $P_\ell$ thus gives us that $b_{st}$
  intersects $\lambda$ in at most two points. We now argue that there
  can actually only be one such a point. Assume by contradiction that
  there are two such points $w_1$ and $w_2$, and that the subcurve
  $\lambda' \subseteq \lambda$ from $w_1$ to $w_2$ lies in the Voronoi
  region of $t$. We now consider two cases, depending on whether
  $\lambda'$ contains a vertex $v$ of $P_\ell$.

  \begin{figure}[tb]
    \centering
    \includegraphics{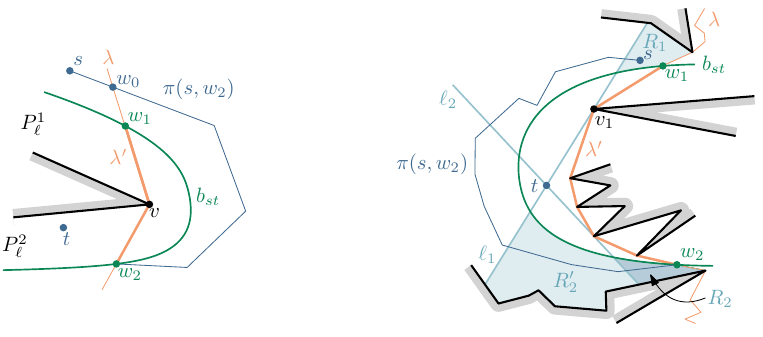}
    \caption{The two cases for the proof that the bisector
      $b_{st}$ can intersect $\lambda$ only once.}
    \label{fig:one_intersection}
  \end{figure}

  In case such a vertex $v$ exists, $v$ actually decomposes $P_\ell$
  into two subpolygons $P_\ell^1$ and $P_\ell^2$ (see
  Figure~\ref{fig:one_intersection}), one of which, say $P_\ell^1$
  contains $s$ and $w_1$, and the other polygon $P_\ell^2$ contains
  $w_2$. The interior of shortest path $\pi(s,w_2)$ is contained in
  the Voronoi region of $s$. However, that means $\pi(s,w_2)$ must
  cross $\lambda$ in $\lambda\setminus \lambda'$, say in a point
  $w_0$. However, as $\lambda$ is a shortest path, that means that we
  can shortcut this path by replacing $\pi(w_0,w_2)$ by a subpath of
  $\lambda$. Contradiction.

  So, we conclude that $\lambda'$ cannot contain any vertices from
  $P_\ell$. Hence $\lambda'$ is a convex chain, whose internal
  vertices (if they exist) are vertices of $P_r$. Consider the two
  tangents of $t$ with this convex chain and clip them so that they
  are chords of $\POld$. In particular, let $\ell_1$ be the chord
  defined by the first point of tangency $v_1$ (considered in the
  order along $\lambda'$) and $\ell_2$ be the second chord. Now
  observe that the part of the bisector $b_{st}$ between $w_1$ and
  $w_2$ must intersect $\ell_1$ twice: once on each side of the point
  of tangency. Moreover, the interval between these points (that thus
  contains $v_1$) must be in the Voronoi region $V_t$ of $t$, and thus
  $\ell_1 \setminus V_t$ consists of two 
  segments $e_2$ and $e_1$ (in that order along $\ell_1$ when it is
  oriented from $t$ to $v_1$).

  This chord $\ell_1$ splits $\POld$ into
  two subpolygons $P'_\ell$ and $P'_r$ (with $P'_\ell$ left of
  $\ell_1$). Since $t \in P'_\ell$, it now follows that $s$ has to
  lie in $P'_r$: if $s \in P'_\ell$, Lemma 29 of Agarwal, Arge, and
  Staals~\cite{dynamic_geodesic_nn_arxiv} tells us that $b_{st}$ can
  intersect $\ell_1$ at most once, thus giving us a contradiction.

  It now follows that $s$ must lie in the region
  $(P_\ell \cap P'_r)\setminus V_t$, which consists of two
  disconnected subregions $R_1$ and $R'_2$. One of which, say $R_1$,
  is incident to a part of $e_1$ and contains $w_1$. The other region
  contains $w_2$. Following exactly the same argument as above for the
  other tangent $\ell_2$ (oriented from its point of tangency towards
  $v$), it follows that $s$ must lie right of this tangent, thereby
  further restricting $R'_2$ into a region $R_2$ incident to $\ell_2$
  and containing $w_2$. See Figure~\ref{fig:one_intersection} for an
  illustration.

  Consider the case $s$ lies in the region $R_1$; the case in which
  $s$ lies in $R_2$ is symmetric. It now follows that the shortest
  path from $s$ to $w_2$ crosses $\ell_1$ twice (as it cannot go
  through the Voronoi region of $t$). Since $\ell_1$ itself is a
  shortest path, we obtain a contradiction.

  We conclude that if $s$ and $t$ lie in $P_\ell$, their bisector
  $b_{st}$ can cross $\lambda$ in only one point $w$. Finally, we
  argue that we can test if $b_{st}$ intersects $\lambda$, and if so,
  find the single intersection point~$w$ in $O(\log^2 m)$ time. This
  actually follows directly from the result of Agarwal, Arge, and
  Staals~\cite{dynamic_geodesic_nn_arxiv}. In particular, from their
  Lemma 40 and the discussion that follows it.
\end{proof}

Fix an endpoint $p$ of $\lambda$, and let $S=\{s_1,..,s_n\}$ be the
sites in $P_\ell$, ordered by increasing distance to $p$. Furthermore,
let $S'=\{t_1,..,t_k\}$ be the subset of sites in $S$ that have a
Voronoi region that intersects $P_r$ (and thus $\lambda$), ordered along $\lambda$ starting with the region containing~$p$. Lemma~\ref{lem:one_intersection} implies that $S'$ is a
subsequence of $S$, and that $\Vor$ is a forest. Moreover, it actually
implies that $\Vor$ can be regarded as an Hamiltonian abstract Voronoi
diagram of the sites in
$S'$~\cite{klein1994hamiltonian_vd}. Therefore, we can use the exactly
the same machinery as in Agarwal, Arge, and
Staals~\cite{dynamic_geodesic_nn_arxiv}, we compute $S'$ in
$O(n(\log n + \log^2 m)$, time and then construct $\Vor$ in additional
$O(n\log^2 m)$ time.

\begin{lemma}
  \label{lem:compute_avd}
  The Voronoi diagram $\Vor $ in $P_r$ is forest. An implicit
  representation \F of this forest can be computed
  in $O(n\log n + n\log^2 m)$ time.
\end{lemma}

\subsection{Point location in $\Vor$}

Let $\Vor$ denote the Voronoi diagram of the sites in $S$ in $P_r$ and
$\F$ the corresponding representation of the forest. To find the nearest neighbor of a point
$q \in P_r$ we must be able to perform point location on the forest
$\F$. When the separator $\lambda$ is a (vertical) line segment,
Agarwal, Arge, and Staals observe that the bisectors pieces in $\F$
are $x$-monotone~\cite{dynamic_geodesic_nn}. This allows them to use
the point location data structure of Edelsbrunner, Guibas, and
Stolfi~\cite{point_location_monotone_subdivision} for monotone
subdivision to answer point location queries in $O(\log n \log m)$
time. Sadly, our bisector pieces are no longer monotone. One approach
would be to partition $P_r$ into regions such that the bisector pieces
are monotone within a region. However, this would incur an additional
$O(\log m)$ factor in the query time. We therefore present a different
point location method here.

\subparagraph{When $\F$ is a tree.} The \emph{centroid} $c$ of a tree
$\Tr$ of size $n$ is a node whose removal decomposes the tree into
subtrees $\Tr_1,..,\Tr_k$ of size at most $\lfloor n/2 \rfloor$
each. The \emph{centroid decomposition} $\Tr_C$ of \Tr is a tree with
a root node representing the centroid $c$ and a child node --also a
centroid decomposition-- for each subtree $\Tr_i$. Such a centroid
decomposition has height $O(\log n)$, linear size, and can be computed
in linear time~\cite{giustina19new_linear_time_algor_centr_decom}.
When \F is a tree, our data structure just stores the centroid
decomposition of $\F$.

To answer a query $q$ we do the following. Observe that the centroid
$c \in \F$ corresponds to a vertex of the Voronoi diagram. Let
$s_1,s_2,s_3$ be the sites that define $c$ in counter clockwise order
around $c$, and let $\Tr_1,\Tr_2,\Tr_3$ be the three maximal subtrees
corresponding to the children of $c$ in the centroid
decomposition. More specifically, let $\Tr_1$ be the maximal subtree
connected to $c$ by the edge that is part of the bisector of the
``other'' two sites, i.e. $s_2$ and $s_3$ in this case. The trees
$\Tr_2$, and $\Tr_3$ are defined analogously. The shortest paths from
$c$ to $s_1$, $s_2$, and $s_3$, decompose $P_r$ into three sub-regions
$P_1,P_2,P_3$, where $P_i$ is the subpolygon opposite to
$\pi(c,s_i)$. See Figure~\ref{fig:tree_point_location}.

\begin{figure}
    \centering
    \includegraphics{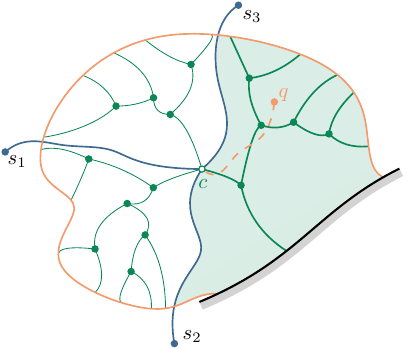}
    \caption{$P_1$ is the green subpolygon.}
    \label{fig:tree_point_location}
\end{figure}

\begin{lemma}
  \label{lem:single_subtree}
  The subtree $\Tr_i$ represents $\Vor \cap P_i$.
\end{lemma}

\begin{proof}
  No edges of $\Vor$ can intersect a shortest path from $c$ to $s_j$,
  because this shortest path is contained within the Voronoi cell of
  $s_j$. So, any subtree of \Vor that that intersects $P_i$ is also
  contained within $P_i$. All that remains to argue is that $P_i$
  intersects (and thus contains) $\Tr_i$. Observe that $P_i$ is
  bounded by edges of the shortest paths from $c$ to the other two
  sites $s_j$ and $s_k$ (with $j \neq k \neq i$). As the Voronoi cells
  of $s_j$ and $s_k$ again contain these paths, it implies their
  bisector must define an edge of \Vor in $P_i$. In particular, this
  is one of the edges incident to $c$. By definition of $\Tr_i$ this
  edge appears in $\Tr_i$, and thus $\Tr_i$ is contained in $P_i$.
\end{proof}

Given a query point $q \in P_r$, our query strategy is now to compute
which subpolygon~$P_i$ contains~$q$. Lemma~\ref{lem:single_subtree}
then guarantees that $q$ lies in the Voronoi region incident to an
edge of $\Tr_i$. We then recursively search in the subtree
$\Tr_i$ until we have a subtree $\Tr'$ of constant size. Each edge of
$\Tr'$ is a piece of a bisector of two sites, and thus gives us at
most two candidate closest sites. For each of the $O(1)$ candidates,
we explicitly query the length of the shortest path in $O(\log m)$
time, and report the closest site.

What remains to describe is how we compute which subpolygon $P_i$
contains $q$. We show that we can do this in $O(\log m)$ time, and
thus this yields an $O(\log n\log m)$ time query algorithm. We first
prove the following helper lemma.

\begin{lemma}
  \label{lem:query_in_subpolygon}
  Given a query point $q \in P_r$, we can test if $O(\log m)$ time if
  $q \in P_i$.
\end{lemma}

\begin{proof}
  Assume without loss of generality that $P_i=P_1$. Consider the
  shortest path $\pi(s_2,c)$, and extend its first and last edge until
  we hit the boundary of $\POld$. Let $\pi(s'_2,c'_2)$ be the
  resulting path. Similarly, extend $\pi(c,s_3)$ into
  $\pi(c'_3,s'_3)$.

  We now observe that if $c$ is a convex vertex of $P_1$ then we have
  that $q \in P_1$ if and only if $q$ lies right of both
  $\pi(s'_2,c'_2)$ and $\pi(c'_3,s'_3)$. If $c$ is a reflex vertex of
  $P_1$ then $q \in P_1$ if and only if $q$ lies right of
  $\pi(s'_2,c'_2)$ or of $\pi(c'_3,s'_3)$. See
  Figure~\ref{fig:def_subpolygon_p1}.

  \begin{figure}[tb]
    \centering
    \includegraphics{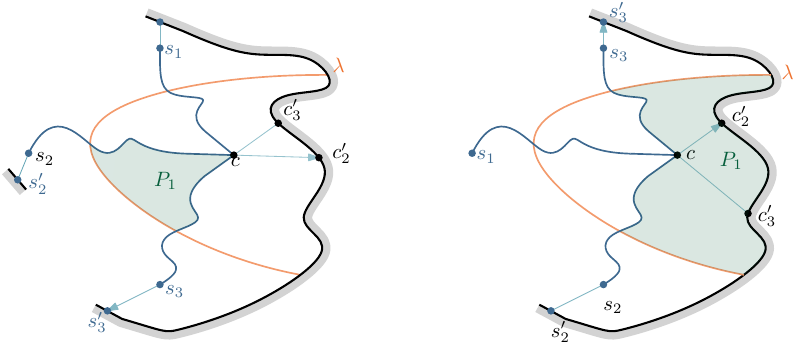}
    \caption{If $c$ is a convex corner of $P_1$ (left) then $q \in
      P_1$ has to lie right of both $\pi(s'_2,c'_2)$ and
      $\pi(c'_3,s'_3)$. If $c$ is a reflex corner of $P_1$ $q \in P_1$
      has to lie right of $\pi(s'_2,c'_2)$ or $\pi(c'_3,s'_3)$.
    }
    \label{fig:def_subpolygon_p1}
  \end{figure}

  We can argue this as follows: observe that $\pi(s'_2,c'_2)$ still
  intersects $\pi(c'_3,s'_3)$ in exactly one point, namely
  $c$. Similarly, $\pi(s'_2,c'_2)$ can intersect $\lambda$ only once;
  in the same points as $\pi(s_2,c_2)$. Therefore, the segment
  $\overline{cc'_2}$ does \emph{not} intersect $\lambda$ or
  $\pi(c'_3,s'_3)$.

  Let $W_{23}$ be the subpolygon enclosed by $\pi(c,s'_2)$, the
  boundary of $P$ in counter clockwise order, and $\pi(s'_3,c)$. We
  thus have $P_1 = W_{23} \cap P_r$. Now if $c$ is convex, it follows
  that the segment $\overline{cc'_2}$ must lie outside of $P_1$. Since
  it does not intersect $\pi(c'_3,s'_3)$ (and analogously
  $\overline{c'_3c}$ does not intersect $\pi(c,s'_2)$), it follows we
  can write $W_{23}$ as the intersection of the two ``right''
  subpolygons defined by $\pi(s'_2,c'_2)$ and $\pi(c'_3,s'_3)$. If $c$
  is reflex, these segments $\overline{cc'_2}$ and $\overline{c'_3c}$
  lie inside $P_1$, and thus $W_{23}$ is the union of these two
  subpolygons.

  Given point $c$ and the sites $s_1,s_2,s_3$, we can easily test in
  $O(\log m)$ time whether $c$ is a convex or a reflex vertex by
  computing the first vertices of $\pi(c,s_2)$ and
  $\pi(c,s_2)$. Similarly, we can compute the endpoints
  $s'_2,c'_2,c'_3,c'_3$ in $O(\log m)$ time by a constant number of ray
  shooting queries. Finally, we then perform the two side test queries
  using Lemma~\ref{lem:side_test_queries}. Since these also take
  $O(\log m)$ time the lemma follows.
\end{proof}

\subparagraph{When \F is a forest.} When \F is a forest, we will add
additional edges between the trees in \F such that we obtain a
single tree $\F^*$ that is still suitable for our point location
queries.

When \F is a forest, each tree in $\F$ has a unique edge that ends
on $\partial P_r$. This means there is a well-defined ordering of the
trees along $\partial P_r$. We obtain the tree $\F^*$ by adding an
edge between between the vertices on $\partial P_r$ of each pair of
neighboring trees. See Figure~\ref{fig:forest_point_location} for an
illustration.

We then proceed as described before in case \Vor is a tree. When
answering a query, it can be that the centroid $c$ is defined by an
endpoint of one of these new edges, i.e. a vertex on $\partial P$. The
shortest paths from $c$ to $s_1$ and $s_2$ again decompose $P_r$ into
three sub-regions: above $\pi(c,s_1)$, between $\pi(c,s_1)$ and
$\pi(c,s_2)$, and below $\pi(c,s_2)$. As before, these correspond
exactly to subtrees of $c$. We can thus determine which subtree to
recurse on by checking whether $q$ lies above or below of both
shortest paths.

\begin{figure}
    \centering
    \includegraphics{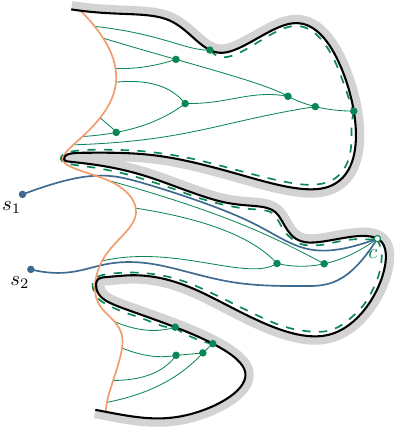}
    \caption{We can extend \F into a tree by connecting the individual
      trees along $\partial P_r$ using the dashed edges.
    }
    \label{fig:forest_point_location}
\end{figure}

\forestBySP*

\section{Concluding remarks}\label{sec:concluding_remarks}
We believe it should be possible to extend our data structure to
support deleting sites as well. In particular, by further
incorporating the results of Agarwal, Arge, and
Staals~\cite{dynamic_geodesic_nn}. However, for ease of exposition we
deferred site deletions for now. Supporting deleting vertices of the
polygon seems more interesting, however also much more
complicated. Deleting a segment might make part of the polygon that was obstructed by this segment before now directly reachable from the query point. It is unclear if there exists a limited number of bounded nearest neighbor queries that can together answer a nearest neighbor query for this area of the polygon. Further ideas are likely required to support this type of deletions as well.

\bibliography{./bibliography.bib}

\end{document}